\documentclass[prl,amsmath,amssymb]{revtex4}
\usepackage{amsmath}
\usepackage{amsthm}
\usepackage{amssymb}
\usepackage{graphics}
\usepackage{graphicx}
\usepackage{feynmf}
\usepackage{color}
\usepackage{subcaption}
\input xy
\xyoption{all} \xyoption{arc}

\newcommand{\be}{\begin{equation}}
\newcommand{\ee}{\end{equation}}

\newtheorem{thm}{Theorem}
\newtheorem{rem}[thm]{Remark}

\newtheorem{lemma}[thm]{Lemma}
\newtheorem{prop}[thm]{Proposition}
\newtheorem{cor}[thm]{Corollary}
\newtheorem{ex}{Example}

\usepackage{bm}

\newcommand{\g}{{\mathfrak g}}
\newcommand{\legr}{{\mathfrak l}}
\newcommand{\rigr}{{\mathfrak r}}

\usepackage{bm}

\begin{document}
\title{A robust generalization of the Legendre transform for QFT} 

\author{D.M. Jackson$^1$, A. Kempf$^2$, A.H. Morales$^3$}
\affiliation{$ $ \\ $^1$Department of Combinatorics and Optimization\\
$^2$Departments of Applied Mathematics and Physics\\
University of Waterloo, Ontario N2L 3G1, Canada,\\
$^3$Department of Mathematics, UCLA, Los Angeles, USA}


\begin{abstract}

Although perturbative quantum field theory is highly successful, it possesses a number of well-known analytic problems, from ultraviolet and infrared divergencies to the divergence of the perturbative expansion itself. As a consequence, it has been difficult, for example, to prove with full rigor that the Legendre transform of the quantum effective action is the generating functional of connected graphs. Here, we give a rigorous proof of this central fact. To this end, we show that the Legendre transform can be re-defined purely combinatorially and that it ultimately reduces to a simple homological relation, the Euler characteristic for tree graphs. This result suggests that, similarly,  also the quantum field theoretic path integral, being a Fourier transform, may be reducible to an underlying purely algebraic structure.  

\end{abstract}

\maketitle

\section{Introduction}
Legendre transforms play important r\^oles in quantum field theory (QFT). The Legendre transform of the quantum effective action is the generating functional of connected graphs, and the Legendre transform of the classical action is the generating functional of tree graphs. That it is the Legendre transform which provides these maps is nontrivial to show. The conventional explanation appeals to the relationship of the Legendre transform to the stationary phase and steepest descent approximations of the occurring integrals, see, e.g., \cite{WeinVol2}.

Interestingly, however, in perturbative QFT the condition for the Legendre transform to apply, namely the convexity of the function to be transformed, is generally not met: even after ultraviolet divergencies have been renormalized, the small-coupling perturbative expansions in QFTs generally lead to divergent expressions, instead of leading to well-defined convex generating functionals. 

As is well known, this problem is deep and arises already in the simple case of the euclidean path integral of scalar QFT, even when reduced to just one degree of freedom: the one-dimensional ordinary integral
$\tilde{Z}[J]:=\int_{-\infty}^\infty e^{-\phi^2-\lambda \phi^4 +J\phi}d\phi$
is finite for all $\lambda\ge 0$ but divergent for all
$\lambda< 0$. Therefore, $\tilde{Z}[J]$, when expanded in powers of $\lambda$ about $\lambda=0$, has a radius of convergence of zero. Nevertheless, in perturbative QFT one constructs just such a perturbative expansion in $\lambda$, namely by 
writing $\tilde{Z}[J]=\int_{-\infty}^\infty e^{-\lambda \partial^4_J}~e^{-\phi^2+J\phi}d\phi$ and then pulling the derivatives out of the integral: ${Z}[J]:=e^{-\lambda \partial^4_J}\int_{-\infty}^\infty e^{-\phi^2+J\phi}d\phi$. The reason why such an expansion can be useful in QFT is that they are usually asymptotic series, i.e., the partial sums approach the true value of $\tilde{Z}[J]$ to some extent before eventually diverging. 

Fundamentally, however, since ${Z}[J]$ is divergent, ${Z}[J]$ is neither convex nor is it even a function or functional. 
Consequently, the perturbatively-defined $iW[J]:= \log{{Z}[J]}$, which in QFT is the generating function of the connected graphs, is not a well-defined function either. But $W[J]$ would have to be a well-defined and also convex function in order to be able to take its Legendre transform to arrive at the effective action. 

The fact that the physical predictions of perturbative quantum field theory are very successful indicates that the equations of the Legendre transform nevertheless do hold true. The reason why they hold true cannot be analytic, however. We take this to indicate that there is a deeper reason, some basic mechanism, to be discovered. Our aim here is to find this underlying reason for why the equations of the Legendre transforms in QFT hold true. 

We need to ensure first that all the quantities in question are well defined. To this end, we begin by showing that the Legendre transform need not be viewed as a map from functions into functions or functionals into functionals. Instead, the Legendre transform can be viewed as mapping the coefficients of one formal power series into the coefficients of another formal power series. Here, the term ``formal" does not express ``mathematically non-rigorous", as it often does in the physics literature. Instead, the term ``formal power series" is here a technical mathematical term, meaning a power series in indeterminates. Formal power series are not functions. A priori, formal power series merely obey the axioms of a ring and questions of convergence do not arise. Formal power series are a key tool, for example, in the field of combinatorial enumeration \cite{jackson}.  

The use of formal power series allows us to construct a purely algebraic Legendre transform for which we  prove rigorously that it transforms the quantum effective action into the generating functional of connected graphs, and that it transforms the classical action into the generating functional of tree graphs. We will thereby discover the underlying reason - which is completely robust against analytic issues - for why the equations of the Legendre transform hold true in perturbative QFT: the Legendre transform, properly defined as a map between the formal power series, ultimately expresses the Euler characteristic of tree graphs.  
We will here build on our previously unpublished preliminary preprint \cite{JKM1}. 
Other frameworks from combinatorics, like the theory of species \cite{BookSpecies} have been used in \cite{lerouxbrydges} and \cite{faris} to also analyze the Legendre and Fourier transforms in QFT.

\section{Setup and Notation}

\noindent At the heart of the path integral formulation of quantum field
theory is the integral over fields. For example, for scalar fields on flat space:
 \begin{equation} \tilde{Z}[J] = 
\int e^{i S[\Phi] +i \int J\Phi~d^rx} D[\Phi] ~\label{one}
\end{equation}
Notice that $\tilde{Z}[J]$ is the Fourier transform of $e^{iS[\Phi]}$, e.g. see \cite[\S 3.4]{ConnesMarcolli}.
Here, $S$ is the classical action, $\Phi$ and $J$ stand for a real bosonic field and its source field, which we assume to be elements of same Hilbert space of square-integrable functions.
We choose the units such that $c=\hbar=1$. We assume suitable ultraviolet and infrared
cutoffs so that the space of fields, equipped with the inner product $\langle
J,\Phi \rangle = \int J(x) \Phi(x) d^rx$, is of finite dimension, say $N$.
Choosing an orthonormal basis, $\{b_a\}_{a=1}^N$, in the space of fields, we
have $\Phi = \Phi_a b_a,~J = J_a b_a,~ \langle J,\Phi\rangle =J_a\Phi_a$, and
\begin{equation} S[\Phi] = \sum_{n\ge 2} \frac{1}{n!}
S^{(n)}_{a_1,...,a_n}\Phi_{a_1}\cdots \Phi_{a_n},~ \label{actionexpansion}
\end{equation}
where twice occurring indices are to be summed over. We
assume that $S$ has no linear term and that $S^{(2)}$ contains a
convergence-inducing Feynman $i\epsilon$ term. The regularized Fourier
transform finally reads: $\tilde{Z}_r[J] = \int_{\mathbb{R}^N} e^{iS[\Phi]
+iJ_a\Phi_a} \prod_j d\Phi_j$. 

The small-coupling perturbative approach to evaluating $\tilde{Z}_r[J]$ is to 
pull the interaction terms before the integral by using derivatives,
completing the squares and carrying out the integrations, to obtain:
\[ Z_r[J] = \mu'\exp\left({\sum_{n>2}
\frac{i}{n!}S^{(n)}_{a_1,...,a_n}\partial_{(iJ_{a_1})}\cdots
\partial_{(iJ_{a_n})}}\right)\cdot \exp\left({(iJ_b)
\frac{i}{2}{S_{bc}^{(2)}}^{-1}(iJ_{c})}\right)\] 
The calculation of these partial sums then proceeds by viewing $Z[J]$ (we are
now dropping the subscript ``$r$") as the generating functional of all Feynman
graphs $\g$ built from the Feynman rules $edge= i (S^{(2)})^{-1}$, and
$n$-$vertex = iS^{(n)}$, with at least one edge. One can view $Z[J]$ also as a
sum of all graphs with the additional Feynman rule $1$-$vertex= iJ_a$, where
each graph $\g$ has a symmetry factor $\omega(\g)=2^{-\ell}k^{-1}$. Here,
$\ell$ is the number of edges of $\g$ joining a vertex with itself and $k$ is
the number of automorphisms of $\g$. Note that if $\g$ is a tree graph with
labelled ends (\it i.e., \rm no unlabelled $1$-vertices) then $\omega(\g)=1$.


Correspondingly, let us denote the sum of only the connected graphs by $iW[J]$.
When exponentiated, it yields the sum of all graphs, {\it i.e.,}
$\exp(iW[J])=Z[J]$, as can be shown combinatorially, see e.g., \cite{WeinVol2}. Notice that in this way, second quantization beautifully extends the sum-over-all-paths picture of first quantization to include the sum over all Feynman graphs. 

$Z[J]$ is divergent and only its first few partial sums in $\lambda$ are
useful. Similarly, $W[J]$ is divergent. Disregarding this fact, one usually proceeds as if $W[J]$ were a convex function of $J$. This would mean that $\varphi_a=\partial W[J]/\partial J_a$ is invertible to obtain 
$J[\varphi]_a=(J[\varphi])_a$. The Legendre transform of $W[J]$ then yields 
$ 
\Gamma[\varphi] = -  J[\varphi]_a\varphi_a + W[J[\varphi]]
$, 
. The nontrivial claim is that $i\Gamma[\varphi]$ is the generating functional of the sum of
$n$-point $1$-particle irreducible (1PI) graphs for $n>2$, and $i\Gamma^{(2)} =
iS^{(2)} +$ $\sum$ ($2$-point 1PI graphs). The effective action $\Gamma[\varphi]$ was introduced perturbatively in \cite{GSW}. The non-perturbative definition through the Legendre transform is attributed independently to B. DeWitt \cite{dewitt} and Jona-Lasinio \cite{jonadewitt}; see \cite[Ch. 16]{WeinVol2}. \rm One thus obtains, overall:
\begin{equation}
\xymatrixcolsep{3pc} \xymatrix{ ~e^{iS[\Phi]~} \ar[r]^{~\rm{Fourier}} & ~Z[J]~
\ar[l]\ar[r]^{\log/\exp}  &  ~iW[J]~ \ar[l]\ar[r]^{\rm{Legendre}} & \ar[l] ~
i\Gamma[\varphi]~. }     \label{steps}
\end{equation}
Among these three key steps in quantum field theory, the first step, the Fourier transform, suffers from the fact that the path integral is difficult to define analytically. In contrast, the second step, the taking of the logarithm, is mathematically well defined and fully understood: the fact that the exponentiation of the generating series of connected graphs yields the generating series of all graphs can be derived on purely algebraic grounds, see, e.g., \cite[Ch. 16]{WeinVol2}. Our concern here is the third step, the Legendre transform. 
It too suffers from not being well defined analytically. Our aim is to re-define that Legendre transform purely algebraically so that we can then prove the third step rigorously. As mentioned in the introduction, we will thereby find the reason why the third step is robust against analytic issues. The reason is that this Legendre transform is ultimately expressing a purely algebraic concept: the Euler relation for tree graphs.


\section{The two r\^oles of the Legendre transform in QFT}

The Legendre transform plays two important r\^oles in quantum field theory. As described in (\ref{steps}), it maps between the generating functional of connected graphs and the quantum effective action. But the Legendre transform also maps between the generating functional of tree graphs and the classical action. Logically, the two r\^oles are equivalent because, as we now briefly review, the quantum effective action can be viewed as an action and the connected graphs can be viewed as tree graphs made of the Feynman rules derived from the effective action. The equivalence of the two r\^oles of the Legendre transform will be useful in the subsequent section because it will allow us to prove one r\^ole with the other then being implied. 

\it The `classical' r\^ole of the Legendre transform in QFT. \rm Assume $T[K]$ is the generating functional of all tree graphs that can be built from the Feynman rules of an action, $F[\Psi]$. 
Assume that the power series $T[K]$ converges to a function which is convex, so that 
the definition $\Psi_a=\partial T[K]/\partial K_a$ is invertible, to obtain $K[\Psi]_a$. It is known that, under these conditions, the Legendre transform of $T[K]$ is the action $F[\Psi]$: 
\begin{equation}
F[\Psi] = -K[\Psi]_a\Psi_a + T[K[\Psi]]. \label{2nd}
\end{equation}
This result is useful, for example, because the perturbative solution to the classical equations
of motion of the field theory can be obtained from the generating functional of  tree graphs: consider the action, $F[\Psi]+\int K\Psi d^rx$, of a
classical system coupled linearly to a source field, i.e., a driving
force, $K$. The equations of motion, $\delta F/\delta\Psi = -K$, are
to be solved for the field $\Psi[K]$ as a functional of the applied
source $K$. From (\ref{2nd}), the inverse Legendre
transform $T[K]=F[\Psi]+\int K\Psi d^rx$ of $F[\Psi]$ then yields the
generating functional, $T[K]$, of trees. From $\Psi[K]= \delta T[K]/\delta K$ then follows that $T[K]$, once differentiated
by $K$, yields the perturbative solution, $\Psi[K]$, to the classical equations
of motion in powers of the perturbing source field $K$.

Below, we will develop a new Legendre transform of the form of (\ref{2nd}) which is much more robust in the sense that it does not require assumptions of convergence and convexity.

\it The `quantum' r\^ole of the Legendre transform in QFT. \rm This is the third step in (\ref{steps}). Any connected
graph can be viewed as consisting of so-called maximal one-particle irreducible (1PI) subgraphs that are connected
by edges whose deletion would disconnect the graph. For practical calculations
of Feynman graphs, this conveniently identifies the 1PI graphs as building
blocks:  After renormalizing them, the 1PI graphs can be glued together to form
connected graphs with no further loop integrations needed. 

Every connected graph is, therefore, a tree graph made out of a new set of Feynman rules: Any connected graph is a tree graph whose vertices are $1$PI graphs connected by strings of edges and $2$-point $1$PI graphs. Thus, the generating functional of connected
graphs, $iW[J]$, can be viewed as the generating functional of all tree graphs made from new Feynman
rules where the $n$-vertex is the sum of all $n$-point $1$PI
graphs, while the edge of the new Feynman rules is given by \begin{equation}
-~+~-\!\!\bigcirc\!\!-~+~-\!\!\bigcirc\!\!-\!\!\bigcirc\!\!-
~+~-\!\!\bigcirc\!\!-\!\!\bigcirc\!\!-\!\!\bigcirc\!\!- ~+~ \dots
~~~=~~~\left((-)^{-1} + \bigcirc\right)^{-1}\end{equation} Here, $\bigcirc$ is the sum of
$2$-point $1$PI graphs and we summed a geometric series. 

By definition, the action which induces these Feynman rules is the quantum effective action, $\Gamma[\varphi]$, see, e.g., \cite{dominicis,GSW,dewitt,jonadewitt,vasiletal}. Therefore, the quantum effective action is that action
which when treated classically (i.e., only to tree level) yields the correct quantum theoretic
answer: on one hand, any $n$-point function can be calculated from first principles as the corresponding sum of all
connected graphs using the Feynman rules of the action $S$. On the other hand, if the quantum effective action is known, then the same $n$-point function can also be calculated as if one solved a classical problem, namely by summing the corresponding tree graphs made out of the Feynman rules generated by the quantum effective action $\Gamma[\varphi]$, see, e.g., \cite{WeinVol2}.

Considering (\ref{2nd}), this means that the quantum effective action is the Legendre transform of the generating functional of connected graphs $W[J]$, if the conditions for the Legendre transform to be applicable are fulfilled. However, within the conventional approach, since $W[J]$ is divergent, it is not a convex functional and the prerequisites for applying the Legendre transform are not met.


\section{Main results}  
Let us now develop an algebraic Legendre transform for which (\ref{2nd}) holds without any analytic prerequistites. This then allows us to apply (\ref{2nd}) to rigorously prove that the algebraic Legendre transform maps between actions and generating series of tree graphs, including what is then the special case of the Legendre transform mapping between the quantum effective action and the generating series of the connected graphs.  

To this end, let us redefine $S,Z,W,\Gamma, T$ and $F$ as elements in a ring of
formal power series. All physically relevant
information is thereby viewed as encoded in the coefficients of the various powers. The question of the convergence of the series does not arise.

We can, therefore, define a \it combinatorial Legendre transform: \rm Assume given a formal power series $F[\Psi]$. We then define its algebraic Legendre transform through the following algebraic procedure: view $F[\Psi]$ as an action, read off the Feynman rules and then obtain $T[K]$ as the formal power series generating all tree graphs. Here, $F[\Psi]$ can be any arbitrary formal power series, subject to the condition that it has no constant or linear term and that the coefficient of its quadratic term must be invertible. This ensures the absence of 0- and 1- vertices and that the Feynman rule for the edge exists. The new combinatorial Legendre transform, which does not require analytic assumptions, obeys the key equation of the conventional analytic Legendre transform, as we first found in our preliminary preprint \cite{JKM1}:

To state our main theorem we need the following notation. Let $F[\Psi]  =  \sum_a  \frac{1}{2}  F_{a,a}^{(2)}  \Psi_a^2 + \sum_{n\ge 3} \frac{1}{n!} 
F^{(n)}_{a_1,...,a_n}\Psi_{a_1}\cdots \Psi_{a_n}$ be an element in a ring of
formal power series with the indeterminates $\Psi_a$ assumed commutative. Further assume that its 
coefficient matrix {$F^{(2)}$} is invertible. (An example is $S$ in (\ref{actionexpansion})). We view
 $F[\Psi]$ as an
action that defines Feynman rules. The corresponding tree graphs will be generated by a formal
power series, in variables $K_a$, which we denote by $T[K]$.

\begin{thm}
\label{thm:main}
Given a formal power series $F[\Psi]$ as above and the corresponding series of trees $T[K]$, we  relate the two indeterminates $K$ and $\Psi$ by defining   $K$ to be the formal power series  
$K[\Psi]_a ~:=~ -\partial F[\Psi]/\partial\Psi_a$. 
Then, the formal
power series $T[K]$ is the Legendre transform of the formal power series $F[\Psi]$ in the sense that these power series obey the Legendre transform equation
\begin{equation}
T[K] = K_a\Psi[K]_a  + F[\Psi[K]]. \label{toshow}
\end{equation}
\end{thm}

\begin{rem} \label{rem:condLT}
\rm 
Note that the composition $F[\Psi[K]]$ is well-defined as a formal power series since
$\Psi[K]$ has no constant term. 
Further, the algebraic derivative $\partial F[\Psi]/\partial\Psi_a$ in the statement of the theorem is a well-defined operation in the ring, so that
$K[\Psi]_a$ is a formal power series in the $\Psi_a$.
We assume that the compositional inverse, the formal power series $\Psi[K]_a$, exists. For the univariate case, this is equivalent to the linear term in the formal power series $K[\Psi]_a$ being invertible. For the multivariate case, it is sufficient if $K[\Psi]_a$ is of the form $\Psi_a\cdot G$ for a multivariate series $G$ with invertible constant term. For more details on the conditions for (multivariate) Lagrange inversion, see \cite[\S 1.1.4, \S 1.2.9]{jackson}. 
\end{rem}

We emphasize that here, unlike for the analytic Legendre transform, for this algebraically-defined Legendre transform it is not necessary that the formal power series $F[\Psi]$ converges to a convex function or that it converges at all.

In the next section we will give examples of Theorem~\ref{thm:main} for {\em univariate} actions $F(x)= - x^2/2 + \sum_{n\geq 3} F^{(n)} x^n/n!$ and formal power series of trees $T(y)=-y^2/2 + \sum_{n\geq 3} T^{(n)} y^n/n!$ based on known examples of trees from enumerative combinatorics. To this end, we now give an explicit formula for $T(y)$ in terms of the formal parameters $F^{(n)}$ for each $n$-vertex. The examples that follow can be obtained by doing different evaluations of these parameters. For examples of formal power series of all graphs of similar univariate actions in the context of asymptotic expansions for large coefficients see \cite[\S 2]{Borinsky}.

\begin{cor} \label{cor:main}
Let $F(x)= - x^2/2 + \sum_{n\geq 3} F^{(n)} x^n/n!$ be a formal power series and $T(y)$ be its Legendre transform then $T(y)=-y^2/2 + \sum_{n\geq 3} T^{(n)} y^n/n!$ where $T^{(n)}$ for $n\geq 3$ is 
\begin{equation} \label{eq:degist}
T^{(n)} = \sum_{n_3,n_4,\ldots,n_k} \frac{(n-2+\sum_{j=3}^k n_j)!}{\prod_{j=3}^k n_j!} \prod_{j= 3}^k \left(\frac{F^{(j)}}{(j-1)!}\right)^{n_j},
\end{equation}
where the sum is over all finite tuples $(n_3,\ldots,n_k)$ of nonnegative integers satisfying the Euler relation $\sum_{j=3}^k (j-2)n_j = n-2$.
\end{cor}

In Section~\ref{sec:proofs}, we give two proofs of this corollary of Theorem~\ref{thm:main}. The first proof uses known explicit formulas for the compositional inverse of ordinary formal power series. The second proof is combinatorial and involves counting the trees directly.

\begin{rem}
\rm Regarding the computational complexity of \eqref{eq:degist}, we notice that the number of terms in \eqref{eq:degist} equals $p(n-2)$ where $p(n)$ is the well-known number of integer partitions of $n$, see \cite[\href{https://oeis.org/A000041}{A000041}]{oeis} ($n_j$ for $j\geq 3$ is the number of parts $j-2$ in the partition). The number $p(n)$ of integer partitions of $n$ has the generating function
$\sum_{n\geq 0}^{\infty} p(n) x^n = \prod_{j=1}^{\infty} 1/(1-x^j)$,
and by the  Hardy--Ramanujan formula $p(n) \sim \frac{1}{4\sqrt{3}n} e^{\pi\sqrt{2n/3}}$ as $n\to \infty$. 
\end{rem}

For the univariate case we denote the algebraic Legendre transform that sends $F(x)$ to $T(y)$  by $(\mathbb{L}\, F)(y)$. Next, we state some properties of $\mathbb{L}$. There are analogous properties for the multivariate case.

\begin{prop} \label{prop:propsLT}
Let $T(y)$ be the Legendre transform of $F(x)$. Then
\begin{itemize}
\item[(i)] $\partial T(y)/\partial_y = x(y)$.
\item[(ii)] $T(y)$ also has a Legendre transform.
\item[(iii)] The Legendre transform $\mathbb{L}: F(x) \mapsto T(y)$ is a quasi-involution, i.e. $(\mathbb{L}^2 F)(-x) = F(x)$. 
\end{itemize}
\end{prop}

\begin{proof}
To prove (i) we differentiate $T(y)$ given by \eqref{toshow} using the product rule and the chain rule 
\begin{align*}
\frac{\partial T(y)}{\partial y} &= \frac{\partial}{\partial y} \left( y \cdot x(y) + F(x(y))\right)\\
&= x(y) + y \frac{\partial x(y)}{y} + \frac{\partial F}{\partial x} \cdot \frac{\partial x(y)}{\partial{y}} = x(y),
\end{align*}
The last equality holds since $\partial F/\partial{x} = -y$.

Next we show (ii). Either from \eqref{toshow} or directly counting nontrivial trees, we see that the series $T(y)$ is of the form $y^2/2 + \sum_{n\geq 3} T^{(n)} y^n/n!$. Thus, $T(y)$ has no constant or linear term and it has a nonzero quadratic term. Therefore, $z := -\partial T(y)/\partial_y = y + \sum_{n\geq 2} T^{(n+1)} y^n/n!$ has a compositional inverse. Thus $T(y)$ has a Legendre transform (see Remark~\ref{rem:condLT}).

Lastly, to prove (iii) let $z = -\partial T/\partial{y}$ be the new variable that by (i) equals $-x(y)$. Thus, the Legendre transform of $\mathbb{L} F$ equals 
\begin{align*}
(\mathbb{L}^2 F)(-x) &= -x\cdot y(-x) + T(y(-x))\\
& = -x\cdot y(-x) + \left(x\cdot y(-x) + F(x)\right) = F(x),
\end{align*}
as desired.
\end{proof}

\begin{rem} \label{rem:remembering}
\rm Notice that Proposition \ref{prop:propsLT}~(ii),(iii) implies that the Legendre transform of a power series that is divergent is not only well defined. It is also not forgetful, in the sense that the original power series can be recovered, in fact by Legendre transforming again. 
\end{rem}

\subsection{Examples of Theorem~\ref{thm:main}}

\begin{ex}[tetravalent trees - convex polynomial action] \label{ex:tetraconv}
The action $F_1(x) = -x^2/2 - x^4/4!$ (i.e. $F^{(4)}=-1$ and $F^{(k)}=0$ for other $k$) is a convex function and it, therefore, possesses an analytic Legendre transform. By solving the cubic equation $y=x+x^3/6$ for $x=x(y)$ we obtain
\[
x(y)=\sqrt [3]{3\,y+\sqrt {8+9\,{y}^{2}}}-{\frac {2}{\sqrt [3]{3\,y+
\sqrt {8+9\,{y}^{2}}}}},
\]
and then $T^{\text{a}}_1(y)$ is defined as
\begin{equation} \label{eq:ex1}
T^{\text{a}}_1(y) = y\cdot x(y) - \frac{x(y)^2}{2}-\frac{x(y)^4}{4!}.
\end{equation}

To calculate it combinatorially, we start by noticing that its trees consist of unlabelled $4$-vertices and labelled $1$-vertices weighted by $(-1)^{n_4}$ where $n_4$ is the number of $4$-vertices in the tree. Using \eqref{eq:degist} or using Lagrange inversion one can show that 
\[
T_1(y) = \sum_{k\geq 1} (-1)^{k-1}\frac{(3k-3)!}{(k-1)! 6^{k-1}}\frac{y^{2k}}{(2k)!} = \frac{{y}^{2}}{2}
\cdot {\mbox{$_3\mathsf{F}_2$}\left(\frac{1}{3},\frac{2}{3},1;\,\frac{3}{2},2;-\,{\frac {9}{8}}\,{y}^{2}\right)},
\]
where $_3\mathsf{F}_2(a_1,a_2,a_3; b_1,b_2; x)$ is a generalized hypergeometric function. 
The radius of convergence of $T_1(y)$ is $2\sqrt{2}/3$. By Theorem~\ref{thm:main}, $T_1(y)$ is the combinatorial Legendre transform of $F_1(x)$ and equals the RHS of \eqref{eq:ex1}. Thus $T_1^{\text{a}}(y)=T_1(y)$.
See Figure~\ref{fig:allexs} for a plot of $F_1(x)$, $T_1(y)$, of the coefficients of each of the series and also plots of the convergent power series of the examples below.
\end{ex}

\begin{figure}
\begin{subfigure}[b]{.25\linewidth}
\includegraphics[scale=0.25]{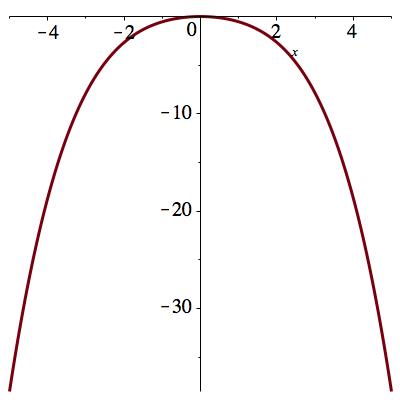}
\caption*{$F_1(x)$}\label{fig:1a}
\end{subfigure}\begin{subfigure}[b]{.25\linewidth}
\centering \includegraphics[scale=0.25]{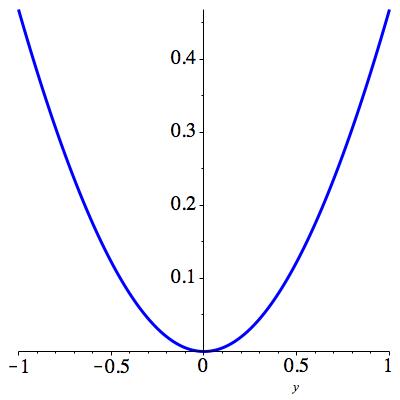}
\caption*{$T_1(y)$}\label{fig:1b}
\end{subfigure}\begin{subfigure}[b]{.25\linewidth}
\centering \includegraphics[scale=0.25]{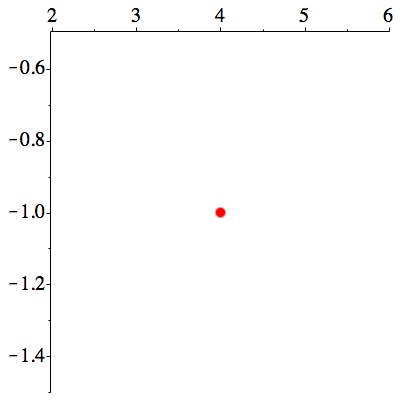}
\caption*{$F_1^{(n)}$}\label{fig:1c}
\end{subfigure}\begin{subfigure}[b]{.25\linewidth}
\centering \includegraphics[scale=0.25]{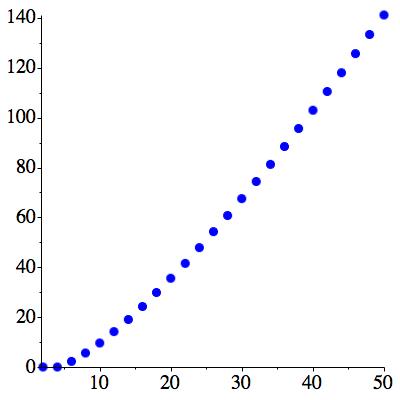}
\caption*{$\ln(T_1^{(n)})$}\label{fig:1d}
\end{subfigure}

\begin{subfigure}[b]{.25\linewidth}
\includegraphics[scale=0.25]{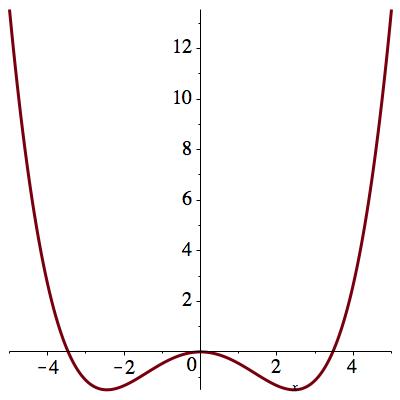}
\caption{$F_2(x)$}\label{fig:2a}
\end{subfigure}\begin{subfigure}[b]{.25\linewidth}
\centering \includegraphics[scale=0.25]{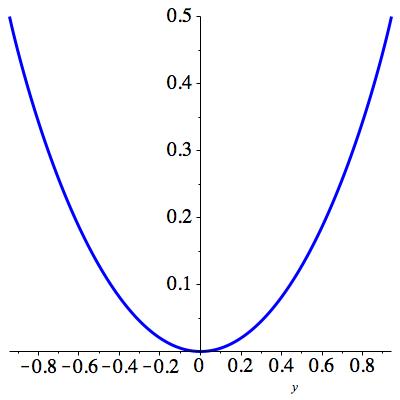}
\caption*{$T_2(y)$}\label{fig:2b}
\end{subfigure}\begin{subfigure}[b]{.25\linewidth}
\centering \includegraphics[scale=0.25]{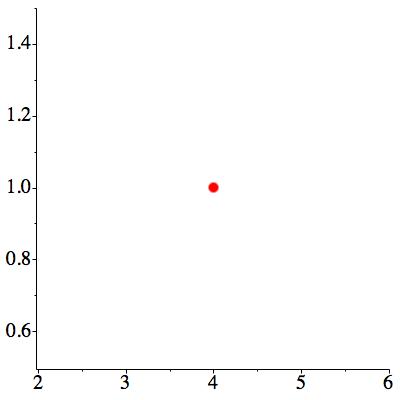}
\caption*{$F_2^{(n)}$}\label{fig:2c}
\end{subfigure}\begin{subfigure}[b]{.25\linewidth}
\centering \includegraphics[scale=0.25]{pt_trans1}
\caption*{$\ln(|T_2^{(n)}|)$}\label{fig:2d}
\end{subfigure}

\begin{subfigure}[b]{.25\linewidth}
\includegraphics[scale=0.25]{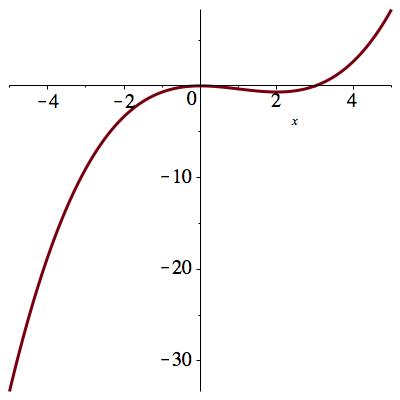}
\caption*{$F_3(x)$}\label{fig:3a}
\end{subfigure}\begin{subfigure}[b]{.25\linewidth}
\centering \includegraphics[scale=0.25]{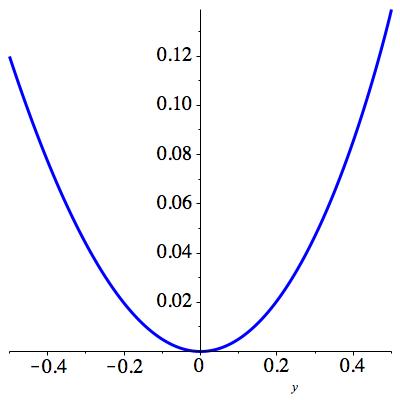}
\caption*{$T_3(y)$}\label{fig:3b}
\end{subfigure}\begin{subfigure}[b]{.25\linewidth}
\centering \includegraphics[scale=0.25]{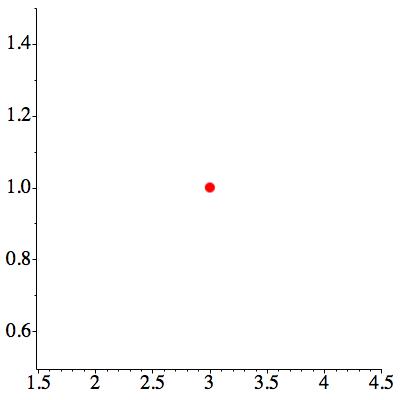}
\caption*{$F_3^{(n)}$}\label{fig:3c}
\end{subfigure}\begin{subfigure}[b]{.25\linewidth}
\centering \includegraphics[scale=0.25]{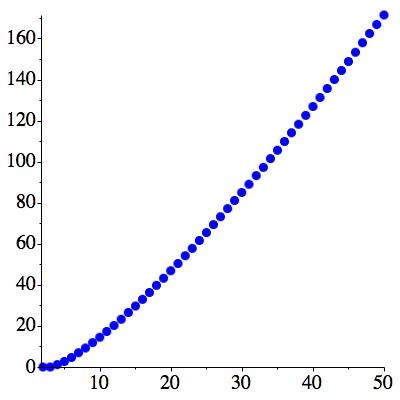}
\caption*{$\ln(T_3^{(n)})$}\label{fig:3d}
\end{subfigure}

\begin{subfigure}[b]{.25\linewidth}
\includegraphics[scale=0.25]{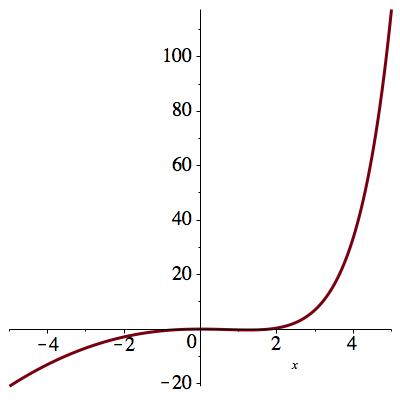}
\caption*{$F_4(x)$}\label{fig:4a}
\end{subfigure}\begin{subfigure}[b]{.25\linewidth}
\centering \includegraphics[scale=0.25]{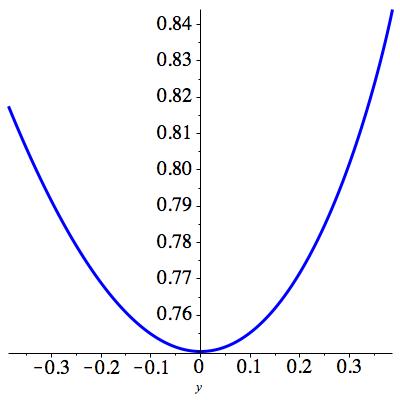}
\caption*{$T_4(y)$}\label{fig:4b}
\end{subfigure}\begin{subfigure}[b]{.25\linewidth}
\centering \includegraphics[scale=0.25]{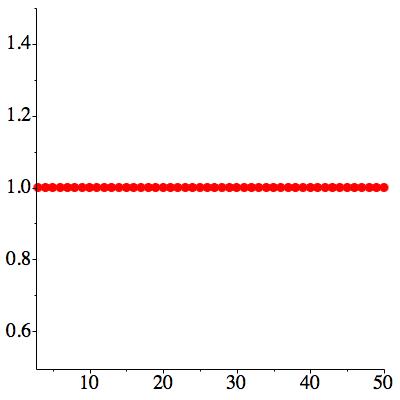}
\caption*{$F_4^{(n)}$}\label{fig:4c}
\end{subfigure}\begin{subfigure}[b]{.25\linewidth}
\centering \includegraphics[scale=0.25]{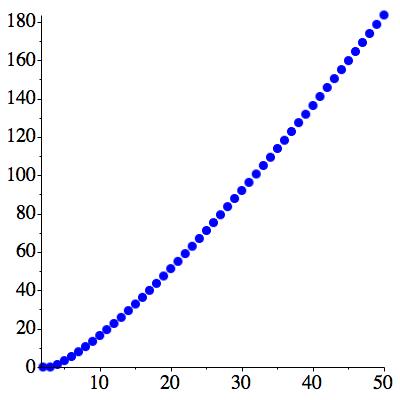}
\caption*{$\ln(T_4^{(n)})$}\label{fig:4d}
\end{subfigure}

\hfill
\begin{subfigure}[b]{.25\linewidth}
\includegraphics[scale=0.25]{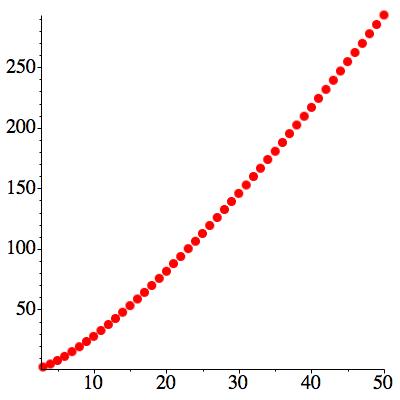}
\caption*{$\ln(F_7^{(n)})$}\label{fig:5c}
\end{subfigure}
\begin{subfigure}[b]{.25\linewidth}
\centering \includegraphics[scale=0.25]{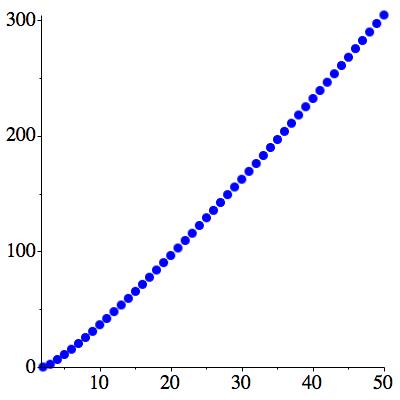}
\caption*{$\ln(T_7^{(n)})$}\label{fig:5d}
\end{subfigure}

\caption{Graphs of the actions $F_i(x)$ and their Legendre transforms $T_i(y)$ (1st and 2nd column), and plots of the ($\ln$ of) coefficients of $F_i(x)$ and $T_i(y)$ in the bases $x^k/k!$ and $y^k/k!$ (3rd and 4th column) for Examples 1, 2, 3, 4, and 7. Example 7 illustrates the important fact that the combinatorial Legendre transform applies even to divergent power series, such as $F_7(x)$ and $T_7(y)$, which cannot be plotted.~~~~~~~~~~~~~~~~~~~~~~~~~~~~~~~~~~~~~~~~~~~~~~~~~~~~~~~~~~~~~~~~~~~~~~~~~~}
\label{fig:allexs}
\end{figure}

\begin{ex}[tetravalent trees - non convex action] 
\label{ex:tetranonconv} 
In contrast to Example~\ref{ex:tetraconv}, the action $F_2(x) = -x^2/2 + x^4/4!$ (i.e. $F^{(4)}=1$ and $F^{(k)}=0$ for other $k$) is not a convex function and, therefore, does not possess an analytic Legendre transform. But we do obtain a unique combinatorial Legendre transform for this action. Its trees consist of unlabelled $4$-vertices and labelled $1$-vertices. Using \eqref{eq:degist} or using Lagrange inversion one can show that 
\[
    T_2(y) = \sum_{k\geq 1} \frac{(3k-3)!}{(k-1)! 6^{k-1}} \frac{y^{2k}}{(2k)!} =\frac{{y}^{2}}{2} \cdot 
{\mbox{$_3\mathsf{F}_2$}\left(\frac{1}{3},\frac{2}{3},1;\,\frac{3}{2},2;\,{\frac {9}{8}}\,{y}^{2}\right)}.
\]
The radius of convergence of $T_2(y)$ is $2\sqrt{2}/3$.
\end{ex}
\begin{ex}[trivalent trees - non convex action] \label{ex:triv} 
Similar to Example~\ref{ex:tetranonconv}, the action $F_3(x) = -x^2/2 + x^3/3!$ (i.e. $F^{(3)}=1$ and $F^{(k)}=0$ for $k>3$) is not convex and, therefore, does not possess an analytic Legendre transform. But we do obtain a unique combinatorial Legendre transform for this action. We start by noticing that the tree graphs stemming from the Feynman rules of this action are trees with unlabelled $3$-vertices and labelled $1$-vertices:\rm
\begin{center}
\includegraphics[scale=0.8]{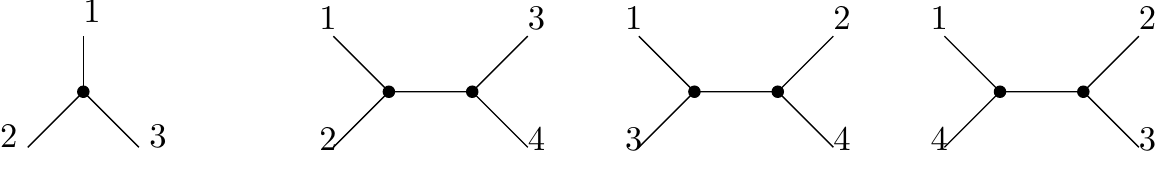}
\end{center}
The number $T_3^{(n)}$ of such trees with $n$ $1$-vertices is given by the {\em double factorial} $(2n-5)!!:=(2n-5)(2n-7)\cdots 3 \cdot 1$. (A trivalent tree with $n-1$ $1$-vertices has $2n-5$ edges each of which can  have a $3$-vertex inserted to yield a tree with $n$ $1$-vertices. This process is reversible.). This can also be seen from \eqref{eq:degist}. The formal power series of these trees is then:
\[
T_3(y) = \frac{y^2}{2} + \sum_{n=3}^{\infty} (2n-5)!!\frac{y^n}{n!} 
\]
The radius of convergence of this power series is $1/2$ and it possesses the closed form representation:
\[
T_3(y) =  \frac{y^2}{2}+{\frac {{y}^{2} \left( -1+3\,y+\sqrt {1-2\,y} \right) }{3 \left( 1
+\sqrt {1-2\,y} \right) ^{2}}}
\]
\end{ex}

\begin{ex} (Nonconvex non-polynomial action convergent on $\mathbb{R}$)
Let us now consider an example of a non-polynomial action, i.e., an action which possesses infinitely many different types of vertices. The tree graphs for the action $F_4(x) = -x^2/2 + \sum_{k=3}^{\infty} x^k/k!=e^x-1-x-x^2$, {\em i.e.} $F^{(k)}=1$ for $k=3,\ldots$, are trees $\mathsf{t}$ with  with unlabelled $k$-vertices for $k=3,\ldots$ and labelled $1$-vertices. The number $T_4^{(n)}$ has no simple formula but it can be computed with the sum in \eqref{eq:degist}. It is known that the formal power series of these trees $T_4(y) = \frac{y^2}{2} + \sum_{k=3}^{\infty} T_5^{(n)} y^n/n!$ has closed form
\[
T_4(y) = - W^2(-e^{(y-1)/2}/2) -2W(-e^{(y-1)/2}/2)  + y^2/4-y/2,
\]
where $W(\cdot)$ is the {\em Lambert $W$ function} (e.g. see \cite[A000311]{oeis} for values of the sequence $(T_4^{(n)})$ and references on these trees). $T_4(y)$ possesses radius of convergence $2\ln(2)-1$. These trees are commonly called {\em phylogenetic trees} since the labelled $1$-vertices can represent biological species.
\end{ex}

\begin{ex} (Nonconvex non-polynomial action convergent with radius of convergence 1)
The tree graphs for the action $F_5(x) = -x^2/2 + \sum_{k=3}^{\infty} x^k/k=-\ln(1-x)-x-x^2$, {\em i.e.} $F^{(k)} = (k-1)!$, are trees $\mathsf{t}$ with unlabelled $k$-vertices for $k=3,\ldots$ and labelled $1$-vertices with a weight $w(\mathsf{t}) = \prod_v (deg(v)-1)!$. The number $T_5^{(n)}$ of such weighted trees with $n$ $1$-vertices is given by $(n-1)!s_{n-2}$ where $s_n$ is the {\em Little Schr\"oder number} \cite[A000311]{oeis}. The formal power series of these trees is then
\[
T_5(y) = \frac{y^2}{2} + \sum_{n=3}^{\infty} s_{n-2} \frac{y^n}{n}, 
\]
and since $(1/y)\frac{d}{dy}T_5(y)$ is the ordinary generating series  of Schr\"oder numbers which has a closed form then
\[
\frac{d}{dy} T_5(y) = (1+y+\sqrt{1-6y+y^2})/4.
\]
\end{ex}

\begin{ex} \label{ex:alex} (Nonconvex non-polynomial action with radius of convergence 1)
The tree graphs for the action $F_6(x) = -x^2/2+\sum_{k=3}^{\infty} x^k/(k(k-1)(k-2))=-(x-1)^2\ln(1-x)/2-x/2+x^2/4$, {\em i.e.} $F^{(k)}=(k-3)!$, are trees $\mathsf{t}$ with unlabelled $k$-vertices for $k=3,\ldots$ and labelled $1$-vertices with a weight $w(\mathsf{t}) = \prod_v (deg(v)-3)!$. Using \eqref{eq:degist} one can show that the  number $T_6^{(n)}$ of such weighted trees with $n$ $1$-vertices is given by $(n-2)^{n-2}$. Thus the formal power series of these trees is 
\[
T_6(y) = \sum_{n\geq 2}^{\infty} (n-2)^{n-2} \frac{y^n}{n!}.
\]
The radius of convergence of this power series is $1/e$ and it possesses a closed form representation
\[
T_6(y) =  \frac{x^2}{2W(-x)} + \frac{x^2}{4W(-x)^2} + x - \frac{1}{4}.
\]
\end{ex}

\begin{ex}[Action with zero radius of convergence]
The tree graphs for the action $F_7(x) = -x^2/2 + \sum_{k\geq 3} (k-1)!x^k$ (i.e. $F^{(k)}=k!(k-1)!$ for all $k\geq 3$) are trees  
 $\mathsf{t}$ with unlabelled $k$-vertices for $k=3,\ldots$ and labelled $1$-vertices with a weight $w(\mathsf{t}) = \prod_v \left(\,deg(v)!(deg(v)-1)!\,\right)$. If $T_7^{(n)}$ is the weighted sum of such trees with $n$ $1$-vertices and $T_7(y)$ is the formal power series $y^2/2 + \sum_{n=3}^{\infty} T_7^{(n)} y^n/n!$. From \eqref{eq:degist} we have that 
\[
T_7^{(n)} = \frac{y^2}{2} + \sum_{n=3}^\infty \left(\sum_{n_3,n_4,\ldots,n_k} \frac{(n-2+\sum_{j=3}^k n_j)!}{\prod_{j=3}^k n_j!} \prod_{j= 3}^k (j!)^{n_j}\right) \frac{y^n}{n!},
\]
where the internal sum is over all finite tuples $(n_3,n_4,\ldots,n_k)$ of nonnegative integers satisfying the Euler relation $\sum_{j=3}^k (j-2)n_j = n-2$. Note that $T_7(y)$ also has zero radius of convergence due the product of factorials $\prod_{j=3}^k (j!)^{n_j}$ in the numerator. 

In this example, the original power series $F_7(x)$ is divergent everywhere, as is its Legendre transform $T_7(y)$. Notice that the Legendre transform even in this case is not forgetful. This is because, as we discussed in Remark~\ref{rem:remembering}, our Legendre transform is a quasi-involution and therefore allows us to fully recover the original power series by applying the Legendre transform again.  
\end{ex}


\begin{rem}
\rm The examples of the actions above belong to two classes of univariate formal power series whose coefficients can be computed efficiently. These are the class of \emph{differentiably finite} series (series $a(x)$ in the ring of formal power series $k[[x]]$ for a field $k$ such that there exist polynomials $p_0(x),\ldots,p_d(x)\neq 0$ in $k[x]$ such that $\sum_{j=0}^d p_j(x) d^j a/dx^j =0$) and the class of so-called \emph{algebraic} series which forms a proper subset (these are series $a \in k[[x]]$ such that there exist polynomials $p_0(x),\ldots,p_d(x)$ in $k[x]$, not all zero, such that $\sum_{j=0}^d p_j(x) a^j =0$); for more details, see \cite[Ch. 6]{EC2}. 

The actions of all of the above examples are differentiably finite and the first three examples are also algebraic.
Differentiably finite power series are not necessarily closed under compositional inverses (take, e.g., $a(x) = \tan(x)$ \cite{Sta}).
Algebraic power series, however, are closed under compositional inverses and also under formal differentiation and integration. This means that 
if an action $F(x)$ is algebraic then also its Legendre transform $T(y)$ is algebraic and therefore efficiently computable. 
\end{rem}

\section{Proofs} \label{sec:proofs}

\subsection{Proof of Theorem~\ref{thm:main}} 

The aim is to prove Theorem~\ref{thm:main}, and in the process to show that (\ref{toshow}) reduces to a basic statement about tree graphs, namely the Euler relation
\begin{equation}
1 ~ = ~ V(\g)~-~E(\g)  ~\label{ec}
\end{equation}
where $V(\g)$ and $E(\g)$ are the number of vertices and edges of the tree graph, respectively. The algebraic nature of the Euler relation, therefore, explains the robustness of the quantum field theoretic Legendre transform with regarding any analytic issues. 

The aim is to show that
\eqref{toshow} is equivalent to \eqref{ec} term by term in the formal power series'. To this end, we now prove that 
equation \eqref{toshow} holds for the coefficients of each of the $m$ - powers of $K$ for all $m\geq 2$, i.e., that 
\begin{equation} \label{eqfivep}
\frac{\partial}{\partial{K_{a_1}}}\dots\frac{\partial}{\partial{K_{a_m}}}~
T\vert_{K=0} = ~
\frac{\partial}{\partial{K_{a_1}}}\dots\frac{\partial}{\partial{K_{a_m}}}
\left(K_a\Psi[K]_a  + F[\Psi[K]]\right) \vert_{K=0}. 
\end{equation}
On the LHS of \eqref{eqfivep} we have, from the definition of $T[K]$, 
the sum of the tree graphs $\g$  which possess $m$ ends. They are labelled by
{}$a_1,\dots,a_m$ and each such $\g$ occurs exactly once, and with the weight
$\omega(\g)=1$. To complete the proof we will show now that the right hand side
consists exactly of all such graphs  with the   multiplicity $V(\g)-E(\g)$. To see this, we write the right hand side of \eqref{toshow}  in  terms  of
tree graphs. We notice that  $K_a=$~$1$-$vertex$,
 and that
\begin{equation}
F[\Psi[K]] ~=~ \sum_{n\ge 2} \frac{1}{n!}~
F^{(n)}_{a_1,...,a_n}~\Psi[K]_{a_1}\cdots~ \Psi[K]_{a_n} \label{F},
\end{equation}
contains $F^{(2)}_{a_1,a_2} ~=~ -(edge)^{-1}$ and
$F^{(n)}_{a_1,...,a_n}~=~n$-$vertex$ for $n>2$. Therefore, the RHS of
\eqref{toshow} takes the form:
\[
(1\mbox{-}vertex)_a~\Psi[K]_{a}~-~\frac{1}{2}~\Psi[K]_{b_1}(edge)_{b_1,b_2}^{-1}~\Psi[K]_{b_2} ~+~ \sum_{n>2} \frac{1}{n!}~
(n\mbox{-}vertex)_{a_1\cdots a_n}~\Psi[K]_{a_1}\cdots~\Psi[K]_{a_n}.
\]
We now differentiate $m$ times and set $K=0$,
to obtain \eqref{toshow}, which in fact contains only graphs which have $m$ labelled ends.
We now recall that, by definition, $\Psi[K]_b~=~\partial
iT[K] / \partial K_b$ is the sum of all trees which have one end vertex
removed. Therefore, the RHS of \eqref{eqfivep} equals the sum of all the tree graphs which possess $m$ labelled ends - obtained by taking one term of  the action, namely \rm $-(\mbox{edge})^{-1}$ or an $n$-vertex, and by attaching the
sum of all the tree graphs to each of its  free  indices. After some
simplification, we obtain that \eqref{2nd} reads, schematically:

\begin{center}
\includegraphics[scale=0.7]{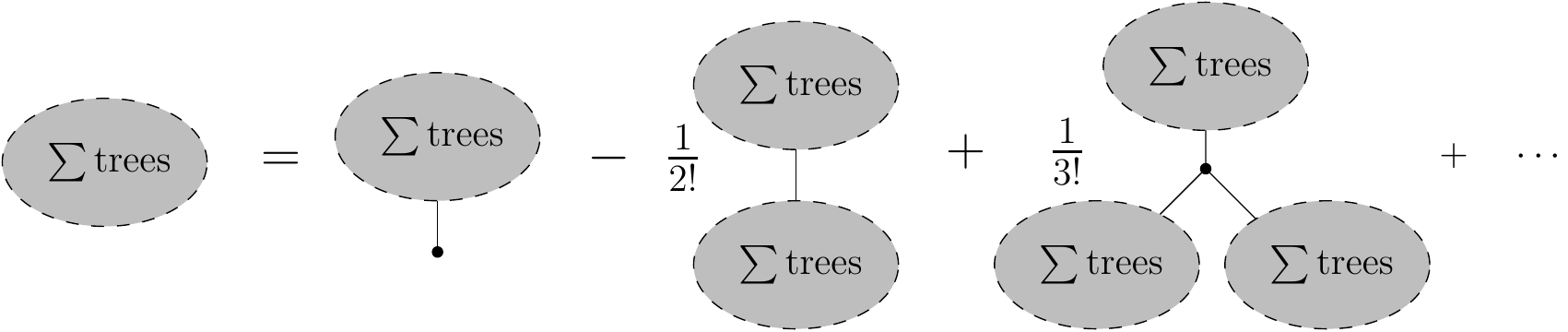}

\end{center}

Let us now consider any tree graph, $\g$, which possesses $m$ labelled
ends. On the LHS, it will occur exactly once. Let us now count this graph's occurrences
on the RHS. To this end, we choose an arbitrary edge $e$  of  $\g$,  and we denote by $\legr$ and $\rigr$ the two subtrees which are to the two sides of
the edge. Then  on the RHS, in the second term the edge $e$ occurs twice,
and because of the $1/2$ in the action, $g$ occurs with weight
$-E(\g)$.

At this point we choose any $n$-vertex, $v$, of $\g$, with $n\in\{1,3,4,...\}$.
We let $\{ {\mathfrak t}_j \}_{j=1}^n $ arbitrarily stand for the sub-trees which emanate from  its  legs. In the remaining terms on the RHS of \eqref{toshow}, the designated vertex, $v$, which has the attached subtrees $\{{\mathfrak t}_j\}_{j=1}^n$ is arising $n!$
times, and this number is cancelled by the $1/n!$ in the action. Therefore, in the remaining terms, $\g$ occurs $V(\g)$ times and it, therefore, occurs with the overall weight $V(\g)-E(\g)$ on the right hand side. This then completes our proof of Theorem~\ref{thm:main}.

\subsection{Algebraic proof of Corollary~\ref{cor:main}}
\noindent To compute the Legendre transform of $F(y)$, we define a new variable $y(x)= -\partial F(x)/\partial x$. This power series has a compositional inverse $y(x)$ has no constant term. We denote this compositional inverse by $x(y)$. By Propostion~\ref{prop:propsLT}~(i) we know that $x = \partial T(y)/\partial y$. Thus, by formally integrating $x(y)$ (i.e., by linearly mapping each monomial $y^n \mapsto y^{n+1}/(n+1)$) we obtain $T(y)$.  

Next we use Loday's explicit formula for compositional inverses of formal power series \cite[\S 6]{loday}, see also the explicit formulas in \cite[\S 2.5]{gessel}.

\begin{lemma}[Loday \cite{loday}] \label{lem:Loday}
If $a(x) = x + a_1x^2 + a_2x^3 + \cdots$ and $b(x) = x+ b_1x^2 + b_2x^3 + \cdots$ is the compositional inverse of $a(x)$, then 
\[
b_n = \sum_{m_1,m_2,\ldots,m_k} (-1)^{\sum_j m_j} B(n+1,m_1,m_2,\ldots,m_k) a_1^{m_1}\cdots a_k^{m_k},
\]
where $B(m_{-1},m_1,m_2,\ldots,m_k)$ is the number unlabelled rooted trees with $m_j$ vertices with $j+1$ children (so $m_{-1}$ are leafs). 
\end{lemma}

The number of trees in the above formula has a product formula. We denote the \emph{multinomial coefficient} by $\binom{N}{n_1,n_2,\ldots,n_k} := N!/(\prod_{j=1}^k m_j!)$ where $n_1+\cdots + n_k = N$.
\begin{lemma}[Sec. 5.3 \cite{EC2} or 2.7.14 \cite{jackson}, Eq. 19 \cite{BernardiMorales}] \label{lem:numBTrees}
\[
B(m_{-1},m_1,\ldots,m_k) = \frac{1}{m_{-1}+m_1+\cdots + m_k} \binom{m_{-1}+m_1+\cdots + m_k}{m_{-1},m_1,\ldots,m_k},
\]
if $m_{-1}+m_1+m_2+\cdots = 1+ 2m_1+3m_2+\cdots$, and $B(m_{-1},m_1,\ldots,m_k)=0$ otherwise.
\end{lemma}

Combining Lemmas~\ref{lem:Loday} and \ref{lem:numBTrees} we have that the compositional inverse of $y = - \partial F/\partial x = x - \sum_{j=3}^{\infty} \frac{F_j}{(j-1)!} x^{j-1}$,
is $x(y) = y + b_1y^2 + b_2y^3+\ldots$, where 
\[
b_n = \sum_{m_1,\ldots,m_k} \frac{1}{n+1+m_1+\cdots + m_k} \binom{n+1+m_1+\cdots + m_k}{n+1,m_1,\ldots,m_k} \prod_{j=3}^k \left(\frac{F^{(j)}}{(j-1)!}\right)^{m_{j-2}},
\]
By formally integrating $x(y)$ with respect to $y$ we obtain $T(y)=y^2/2 + \sum_{j=3}^{\infty} T^{(j)} y^j/j!$. Comparing coefficients of $y^n/n!$  we obtain that $T^{(n)} = (n-1)!\,b_{n-2}$. Thus relabeling $m_0=n-1$, $m_{j-2}=n_{j}$ for $j\geq 3$ and simplifying we get
\[
T^{(n)} = (n-1)! \,b_{n-2}=\sum_{n_3,n_4,\ldots,n_k} \frac{(n-2+\sum_{j=3}^k n_j)!}{\prod_{j=3}^k n_j!} \prod_{j= 3}^k \left(\frac{F^{(j)}}{(j-1)!}\right)^{n_j},
\]
as desired.

\subsection{Combinatorial proof of Corollary~\ref{cor:main}}

The tree graphs for the generic action $F(x)$ are trees with labelled $1$-vertices and unlabelled $k$-vertices for $k=3,\ldots$:
\begin{center}
\includegraphics[scale=0.8]{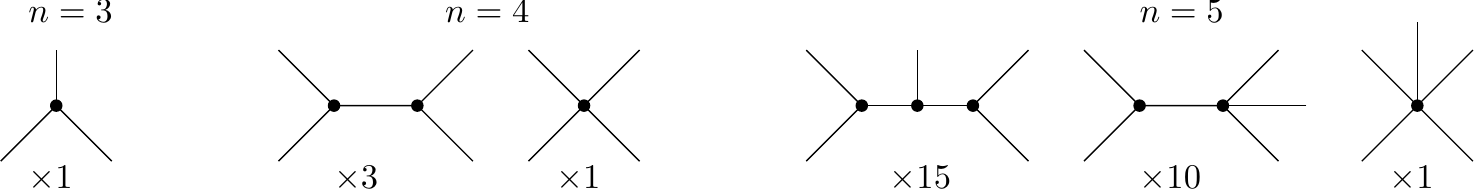}
\end{center}
where each tree $\mathsf{t}$ has a weight $w(\mathsf{t}) = \prod_v F^{(deg(v))}$. 
By Theorem~\ref{thm:main} the Legendre transform $T(y)$ is the formal power series of these trees. The coefficient $T^{(n)}$ of $y^n/n!$ for $n\geq 3$ of $T(y)$ equals
\begin{align}
T^{(n)} &= \sum_{\mathsf{t}} w(\mathsf{t}) = \sum_{n_3,n_4,\ldots,n_k} A(n,n_3,n_4,\ldots,n_k) \prod_{j\geq 3} (F^{(j)})^{n_j}, \label{eq:ATrees}
\end{align}
where $A(n,n_3,n_4,\ldots,n_k)$ is the number of trees with $n$ labelled $1$-vertices, $n_3, n_4,\ldots$ unlabelled $3$-vertices, $4$-vertices, \ldots. Note that the nonnegative integers $n,n_3,n_4,\ldots$ are subject to the Euler relation 
\begin{equation} \label{eq:euler_relAtrees}
\sum_{j= 3}^k (j-2)n_j=n-2.
\end{equation}
The number of trees $A(n,n_3,n_4,\ldots,n_k)$ has an explicit product formula.
\begin{lemma} \label{lem:numATrees}
\[
A(n,n_3,n_4,\ldots) = \frac{(n-2+\sum_{j\geq 3} n_j)!}{\prod_{j\geq 3} n_j!} \prod_{j\geq 3}\left(\frac{1}{(j-1)!}\right)^{n_j},
\]
provided \eqref{eq:euler_relAtrees} holds and $A(n,n_3,n_4,\cdots)=0$ otherwise.
\end{lemma}

Corollary~\ref{cor:main} then follows by using Lemma~\ref{lem:numATrees} in \eqref{eq:ATrees}.

\begin{proof}[Proof of Lemma~\ref{lem:numATrees}]
The result follows by combining \eqref{eq:doublecount} and Lemma~\ref{lem:numBTrees} with the relation
\begin{equation} \label{eq:doublecount}
A(n+2,n_3,n_4,\ldots,n_k) \cdot \prod_{j=3}^k ((j-1)!)^{n_j} = B(n+1,n_3,n_4,\ldots,n_k)\cdot (n+1)!,
\end{equation}
which we now prove by double counting. The RHS of \eqref{eq:doublecount} counts unlabelled rooted trees with $n+1$ $1$-vertices labelled $\{1,2,\ldots,n+1\}$. We replace the root by a $1$-vertex labelled $n+2$ to obtain a {\em plane} tree with $n+2$ labelled $1$-vertices. This last step is clearly reversible.

The LHS counts the same plane trees in the following classical way (see e.g. \cite{MT}). The plane tree is uniquely determined by the underlying tree with $n+2$ labelled $1$-vertices and unlabelled $j$-vertices for $j=3,\ldots$ and the cyclic ordering of size $(j-1)!$ of each of the $j$-vertices.   
\end{proof}

\begin{ex}
We illustrate \eqref{eq:doublecount} when $n_3=n$ and $0=n_4=n_5=\cdots$, from Example~\ref{ex:triv}, we know that $A(n+2,n-1) = (2n-1)!!$. It is well-known that the number of unlabelled rooted binary trees with $n+1$ leafs is given by the $n$th Catalan number $B(n+1,0,n)=C_n := \frac{1}{n+1}\binom{2n}{n}$ \cite[6.19 (d)]{EC2}. 
\begin{center}
\includegraphics[scale=0.8]{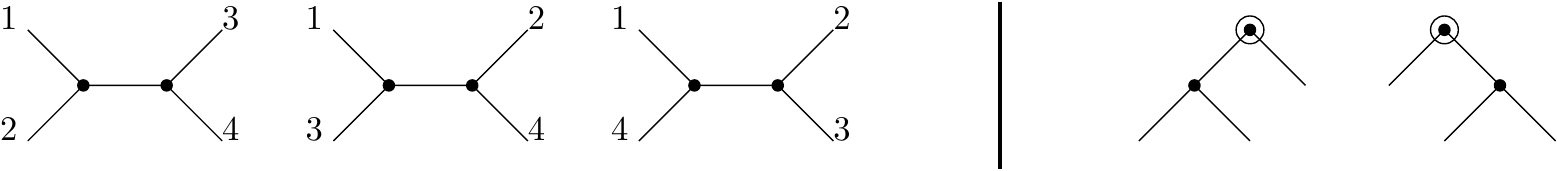}
\end{center}

One can check directly or by double counting plane binary trees with labelled $1$-vertices that
\[
(2n-1)!! \cdot 2^{n} = C_n \cdot (n+1)!.
\]

\end{ex}

\section{Outlook} 

We here defined the generating functionals of all graphs, connected graphs, tree graphs, as well as the action and quantum effective action as formal power series. This expresses the idea that all physical information is encoded in the coefficients of these power series, irrespective of whether or not these series are convergent and therefore also irrespective of whether or not they possess an interpretation as functions. Based on this shift of view, we constructed an algebraic/combinatorial Legendre transform that maps in between the coefficients of such formal power series. 

This then allowed us to rigorously prove the key equation, (\ref{toshow}), of the Legendre transform in quantum field theory. It is the equation that implies that the Legendre transform maps the classical action into the generating functional of tree graphs and that it maps the quantum effective action into the generating functional of connected graphs. The algebraic approach also allowed us to uncover the underlying reason for the robustness of this key equation in regards to analytic issues: (\ref{toshow}) ultimately expresses an Euler relation for graphs, which is a purely combinatorial concept. 

With this, two of the three steps in (\ref{steps}) can be understood purely algebraically: the Legendre transform and the exponentiation that maps the generating functional of connected graphs into the generating functional of all Feynman graphs.  

This suggests that also the remaining step in (\ref{steps}), the Fourier transform, should be understandable purely algebraically and combinatorially. Putting the Legendre transform on solid algebraic footing yielded the insight that the Legendre transform in QFT ultimately expresses an Euler relation. The Fourier transform, re-formulated in algebraic terms, may analogously lead to interesting new insights when applied to the quantum field theoretic path integral, for example, concerning the non-analyticity of the small coupling expansion, see \it e.g., \rm \cite{klauder}. Also, for example the origin of anomalies, conventially understood via Fujikawa's method as originating in the measure of the path integral, may acquire a new algebraic interpretation. 

First results in this direction recently appeared, see \cite{intbydiff1,intbydiff2}, yielding new methods for integration and for integral transforms such as the Fourier and Laplace transforms. An algebraic and combinatorial treatment of the Fourier transform will appear in \cite{JKM2}. Borinsky \cite[\S 2, 4]{Borinsky} recently gave an algebraic definition of the Fourier transform for univariate actions. 

Let us also note in passing that, since the Fourier and Legendre transforms are quasi-involutions, the generating functionals of tree graphs and of connected graphs can themselves be viewed as actions. This expresses a duality between these generating functionals of Feynman graphs on one hand and classical and quantum effective actions on the other.

\begin{rem}
\rm 
As we mentioned before, the present paper is building on the connection we previously made in \cite{JKM1} between the compositional inverse contained in the Legendre transform of an exponential power series and tree graphs. While completing the present paper, it has come to our attention that this connection (without mentioning the Legendre transform) to tree graphs was recently also made by Engbers, Galvin and Smyth \cite[Prop. 3.2]{galvin} where they in addition gave an explicit formula for the compositional inverse formula for the univariate case in terms of trees, amounting to corollary 3. For the exact relationship between the Legendre transform and the compositional inverse, see Prop.~\ref{prop:propsLT} (i).
\end{rem}

\smallskip

\noindent \bf Acknowledgement. \rm We thank Igor Pak and Karen Yeats for interesting discussions on the subject. We thank David Galvin for pointing out \cite{galvin}, Alexander Postnikov for suggesting Example~\ref{ex:alex}, and our referees for helpful comments and suggestions. AK and DMJ acknowledge
support by the Discovery program of NSERC. AHM was partially supported by an AMS-Simons Travel grant.


\end{document}